\documentclass{article}
\pdfoutput=1
\usepackage[utf8]{inputenc}
\usepackage[x11names,dvipsnames,svgnames,table]{xcolor}

\usepackage[export]{adjustbox}
\usepackage{afterpage}
\usepackage{subcaption}
\usepackage{graphicx}
\usepackage{placeins}
\usepackage{pdfpages}
\usepackage{algorithm2e}
\usepackage{array}
\usepackage{booktabs}
\usepackage[most]{tcolorbox}
\usepackage{calligra}
\usepackage{caption}
\usepackage{datetime}
\usepackage{dblfnote}
\usepackage{dirtytalk}
\usepackage{dsfont}
\usepackage{etex}
\usepackage{fancyhdr}
\usepackage{fix-cm}
\usepackage[T1]{fontenc}
\usepackage{textcomp,gensymb} 
\usepackage{graphicx}
\usepackage{lipsum}
\usepackage{listings}
\usepackage{transparent}
\usepackage[everyline=true,framemethod=tikz]{mdframed}
\usepackage{mparhack}
\usepackage{multicol}
\usepackage{multirow}
\usepackage[normalem]{ulem}
\usepackage{verbatim}

\usepackage{algorithmic}

\usepackage[titletoc,title,header]{appendix} 

\usepackage{csquotes}

\usepackage[a4paper]{geometry}
\usepackage{lastpage} 

\usepackage[hidelinks]{hyperref}
\usepackage{cleveref}

\usepackage{amsmath}
\usepackage{amsthm}
\usepackage{amssymb}
\usepackage{amsfonts}
\usepackage{array}
\usepackage{mathtools}

\usepackage{tikz}
\usetikzlibrary{decorations.markings,arrows,arrows.meta}
\usetikzlibrary{positioning}

\usepackage{freetikz}

\tikzset{
  myblock/.style={
    draw,text width=0.75cm,minimum height=2cm,align=center
  },
  mysquare/.style={
    draw,minimum size=1cm,align=center
  },
  twoarrows/.style n args={4}{
    decoration={
      markings,
      mark=at position #1 with {\arrow{>}\node[above] {#3};}, 
      mark=at position #2 with {\arrow{>}\node[above] {#4};} 
    },
  postaction=decorate  
  },
  onearrow/.style 2 args={
    decoration={
      markings,
      mark=at position #1 with {\arrow{>}\node[above] {#2};}, 
    },
  postaction=decorate  
  },
}

\usepackage{enumitem}
\setlist[description]{labelindent=1em}

\newcommand{\bfunc}[1]{\operatorname{\mathtt{#1}}}

\theoremstyle{definition}

\newtheorem{lemma}{Lemma} 
\newtheorem{theorem}{Theorem} 
 
\newtheorem{definition}{Definition}

\title{Specification of the Giskard Consensus Protocol}
\author{Elaine Li$^1$
\and
Karl Palmskog$^2$
\and
Mircea Sebe$^1$
\and
Grigore Ro{\c s}u$^1$}

\date{%
  \normalsize
  $^1$Runtime Verification, Inc., Urbana, IL, USA\\[-0.1cm]
  \texttt{\{elaine.li,mircea.sebe,grigore.rosu\}@runtimeverification.com}\\[0.1cm]
  $^2$KTH Royal Institute of Technology, Stockholm, Sweden\\[-0.1cm]
  \texttt{palmskog@acm.org}
}

\begin{document}
\maketitle

\begin{abstract}
The Giskard consensus protocol is used to validate transactions and computations in the PlatON network. In this paper, we provide a rigorous specification of Giskard, suitable to serve as a reference in protocol implementation and in formal verification. Using our specification, we prove that the protocol guarantees several notable safety properties.
\end{abstract}

\section{Introduction} 
The PlatON network provides a platform for distributed transactions and computations~\cite{PlatON}. The network relies on the Giskard consensus protocol~\cite{Giskard} to reach agreement among participating nodes regarding which blocks of transactions to add to a distributed ledger (blockchain). The protocol tolerates that up to one third of all participating nodes are Byzantine, i.e., may behave adversarially~\cite{Byzantine}.

In this paper, we present a rigorous specification of Giskard. In contrast to the protocol documentation~\cite{Giskard}, which gives an overview of the protocol and describes its general behavior, our aim is to specify the protocol in sufficient detail to serve as a guide for implementation and formal verification. To this end, we provide definitions of both protocol execution semantics and important protocol properties, and present proofs of its three key safety properties based on these definitions.
Since our focus is on protocol safety, we elide discussion of higher-level features of Giskard and only provide brief background; additional context and motivation can be found in the official documentation~\cite{Giskard}.

\section{Background}

Fully asynchronous communication severely restricts the guarantees a consensus protocol can provide~\cite{FLP}. Like the PBFT protocol~\cite{PBFT}, Giskard therefore assumes a partially synchronized mesh network~\cite{CPPS} which models the weak environment of public networks. This means that every network node is in principle able to communicate with every other node by message passing. Similarly to the HotStuff potocol~\cite{HotStuff}, Giskard separates production and confirmation of blocks, and requires blocks to pass multiple \emph{stages} before being confirmed as part to the blockchain.

\section{Participating Nodes}
Giskard is parameterized by a finite set $\mathcal{N}$ of participating nodes. These nodes can play two roles in the protocol: \emph{block proposers} and \emph{validators}. Block proposers are responsible for both producing new blocks and voting on blocks, while validators are only responsible for voting on proposed blocks.

\section{Epochs and Views} 
Giskard uses abstract units of time called epochs. The epoch changes once every 250 blocks~\cite{Giskard}. More precisely, epochs are a measure of time and progress in the protocol tracked by each participating node using an epoch number $e_n$ that is defined based on its local state. For every epoch $i$, an ordered list of $k$ nodes $N_i \subset \mathcal{N}$ is selected to participate in epoch $i$ by epoch $i-1$, using a Verifiable Random Function (VRF).
For each participating node, an epoch is defined as the duration of time required for $10*k$ blocks to reach local Prepare stage, $10*k-1$ blocks to reach local Precommit stage, and  $10*k-2$ blocks to reach local Commit stage. Consequently, epoch change happens asynchronously among participating nodes and depends on the status of their local blocks.

Epochs are further divided into views, tracked using a view number $v_n$. Similar to epochs, view number is defined for each participating node $n$ based on its local state, and the protocol does not guarantee that view change occurs at the same time for all participating nodes. Instead, each participating node updates its local view number $v_n$ depending on its local state. There are two ways in which a participating node can progress onto the next view: 
\begin{itemize}
    \item In the case of a normal view change: all 10 blocks produced by the block proposer reaches quorum locally, or
    \item In the case of an abnormal view change, i.e. a timeout: system timeout occurs, and nodes exchange messages and agree to increment their respective views. The timeout duration for each view is computed by an algorithm, and is known upfront to all participating nodes.
\end{itemize}

Views demarcate the duration of time in which a unique participating node acts as block proposer. Within each epoch, participating nodes from $N_i$ take turns to act as block proposer. Each epoch contains at least $k$ views; each participating node acts as block proposer at least once. If each participant $n \in N_i$ has acted as block proposer once and fewer than $10*k$ blocks reach global Prepare stage, the role of block proposer returns to the front of the list and continues to rotate until $10*k$ blocks reach global Prepare stage, marking the end of the epoch. 
\begin{definition}[Block proposer] 
Let $N_i$ be the list of $k$ participating nodes for epoch $i$. For node $n \in N_i$ with view $v_n = j$, the block proposer for view $j$ is $N_i[j~mod~k]$, i.e., the $(j~mod~k)$-th element of $N_i$. 
\end{definition}

One consequence of this definition is that if two participating nodes agree on the current view number, they agree on the identity of the block proposer. Another consequence is that nodes are always able to verify whether blocks are produced by the correct block proposer in a given view. Therefore, Giskard precludes malicious behavior such as nodes impersonating the block proposer and producing blocks out-of-turn, but does not preclude that the block proposer for a given view behaves maliciously, by, e.g., producing two blocks of the same height.

\section{Messages} 
Participating nodes communicate with one another by broadcasting messages over the network. There are two kinds of messages: consensus messages and synchronization messages. 

\subsection{Consensus messages} 
Consensus messages are the primary mechanism through which Giskard achieves consensus, and are used for block proposing, voting, and performing view changes. Consensus messages are always broadcast to all participating nodes from the sender. There are five kinds of consensus messages: \texttt{PrepareBlock}, \texttt{PrepareVote}, \texttt{ViewChange}, \texttt{PrepareQC} and \texttt{ViewChangeQC}. The following basic data are included in all consensus messages: 


\begin{description}
    \item[\texttt{e}]: the current epoch number
    \item[\texttt{v}]: the current view number
    \item[\texttt{n}]: the node index of the message sender
    \item[\texttt{b\_h}]: a block hash
    \item[\texttt{b\_n}]: the block number
    \item[\texttt{b\_i}]: the block index
    \item[\texttt{sig}]: a Boneh-Lynn-Shacham (BLS) aggregate signature
\end{description}

The current epoch and view numbers indicate when the message was produced. Messages are signed with the message sender's unique identifier, which in addition to the BLS signature allows the receiver to verify that message sender is authentic. As a convention, we use $n$ in place of $N_i[n]$ to refer to the participating node in $N_i$ with index $n$. All consensus messages carry blocks, which are identifiable by their block hash. The block number refers to the height of the block. The block index is view-specific and indicates the block's position in the sequence of blocks produced in that view. Nodes can identify whether a block is the last block using the block index: a block is the last block iff its index is 10. 

Because Giskard's safety properties primarily concern message blocks, views and senders, we abbreviate messages as, e.g., \texttt{PrepareBlock(b,v,n)}, omitting the remaining data fields. 

In addition to the above basic data, consensus messages also contain aggregated BLS signatures of other messages, and therefore can be interpreted as ``carrying'' other messages. The different kinds of message carrying behavior and a discussion of their consequences follow. 

\paragraph*{PrepareBlock.}
\texttt{PrepareBlock} messages contain the following BLS signatures: 
\begin{itemize} 
    \item \texttt{ParentQC}: the aggregated signature of quorum parent block \texttt{PrepareVote} messages, 
    \item \texttt{ViewChangeQC}: the aggregated signature of quorum \texttt {ViewChange} messages. 
\end{itemize} 
Therefore, \texttt{PrepareBlock(b,v,n)} can be viewed as a triple of messages: \texttt{PrepareBlock(b,v,n)}, \texttt{PrepareQC(parent(b),v’,n)}, and \texttt{ViewChangeQC(b\_max, v-1, n)}, where $v' \leq v$, $parent(b)$ is the parent block of $b$, and $parent(b)$ received quorum votes in view $v'$.

However, it is important to note that since new proposers produce blocks in a ``pipeline'' at the beginning of each view, i.e. one after another, the \texttt{ParentQC} field is somewhat of a misnomer for most blocks: It cannot both be the case that all \texttt{PrepareBlock} messages contain the \texttt{PrepareQC} signature of their parent block aggregated from quorum \texttt{PrepareVote} messages, and all \texttt{PrepareBlock} messages in a given view are produced at once, because the messages required for the aggregate do not yet exist for any blocks except the first. Therefore, the \texttt{ParentQC} field for all messages containing non-first blocks is effectively empty, i.e. set to a null placeholder. 


Similarly, the \texttt{ViewChangeQC} field only contains the signature of a \texttt{ViewChangeQC} message when the view change from $v-1$ to $v$ occurred via a timeout. In the case of a normal view change from $v-1$ to $v$, there are no \texttt{ViewChangeQC} or \texttt{ViewChange} messages broadcast amongst nodes. A description of the view change process is contained in the section below. 

\paragraph*{PrepareVote.}
\texttt{PrepareVote} messages contain the following BLS signature: 
\begin{itemize} 
    \item \texttt{ParentQC}: the aggregated signature of quorum parent block PrepareVote messages, 
\end{itemize} 
\texttt{PrepareVote(b,v,n)} can be interpreted as a pair of messages, namely, \texttt{PrepareVote(b,v,n)} and \texttt{PrepareQC(parent(b),v', n)}, where $v' \leq v$, $parent(b)$ is the parent block of $b$, and $parent(b)$ received quorum votes in view $v'$.


One consequence of \texttt{PrepareVote}'s message-carrying behavior is that if any node $n$ receives a \texttt{PrepareVote(b,v,\_)}, it can immediately reciprocate and send \texttt{PrepareVote(b,v,n)}. This is because the node ``extracts'' from the message a \texttt{PrepareQC(parent(b),\_,\_)} message, thus satisfying the criterion that nodes can only vote for blocks whose parents have reached the prepare stage.

To push this consequence even further, it is not the case that nodes necessarily vote for blocks in consecutive order, i.e. sending a \texttt{PrepareVote(b,v,n)} does not imply that $n$ has sent a \texttt{PrepareVote(parent(b),v,n)}. This is illustrated in the example below. 

\paragraph*{Example 1.} 
Let $A$, $B$, $C$ and $D$ be four nodes participating in the protocol, where blocks $b_1$, $b_2$ and $b_3$ have been proposed. $A$, $B$ and $C$ receive \texttt{PrepareBlock} messages for $b_1$, and broadcast \texttt{PrepareVote} messages for $b_1$. $A$, $B$ and $C$ then receive one another's \texttt{PrepareVote} messages for $b_1$, and respectively witness quorums for $b_1$, broadcasting \texttt{PrepareQC} for $b_1$ and \texttt{PrepareVote} for $b_2$. $A$ receives \texttt{PrepareVote(b2, B)} and \texttt{PrepareVote(b2, C)} first, and alongside its own \texttt{PrepareVote(b2, A)}, it witnesses a quorum for $b_2$ and broadcasts \texttt{PrepareQC (b2, A)}.
In the case that $D$ processes \texttt{PrepareQC(b2, A)} first, block $b_2$ reaches prepare stage in $D$'s local state, but $b_2$'s parent block, $b_1$, has not reached prepare stage in $D$'s local state, nor has $D$ itself voted for either $b_1$ or $b_2$.

However, it \emph{is} the case that all blocks prior to $b_n$ must have reached prepare stage globally in order for a participating to send \texttt{PrepareVote(b\_n,v,n)}, meaning that there must be evidence in the global out message buffer witnessing their prepare stage status, either in the form of quorum \texttt{PrepareVote} messages or a \texttt{PrepareQC} message. This illustrates that Giskard's voting criterion concerns what messages have been received, but not what messages have been sent.

\paragraph*{ViewChange.}
\texttt{ViewChange} messages contain the following BLS signature: 
\begin{itemize} 
    \item \texttt{PrepareQC}: the aggregated signature of quorum \texttt{PrepareVote} messages for the block contained in the message.
\end{itemize} 
\texttt{ViewChange(max\_b,v,n)} can  therefore be interpreted as a message pair \texttt{ViewChange(max\_b,v,n), PrepareQC(max\_b,v',n)}. Note that as it is not guaranteed that \texttt{max\_b} reached prepare stage in the present view \texttt{v}, \texttt{v'=v} is not guaranteed. 

In the following sections, we describe transmission of ``carried'' messages separately, for clarity. 

\section{States} 
Each node maintains its own \textit{ViewState} $s$, i.e., a snapshot of the current protocol state. The \textit{ViewState} contains the following basic information: 

\begin{description}
    \item[$e$]: the current epoch number
    \item[$v_n$]: the current view number
    \item[$n_i$]: the node's unique identifier
    \item[$l_{in}$]: the input message buffer
    \item[$l_{pending}$]: the pending message buffer
    \item[$l_{counting}$]: the processed message buffer
    \item[$l_{out}$]: the output message buffer
    \item[$t_n$]: a local clock
\end{description}
The local clock is synchronized across all participating nodes. 

\section{Blocks}
Participating nodes in Giskard ultimately aim to achieve distributed consensus on a blockchain, namely a sequence of blocks containing domain-specific data, such as transactions. Each participating node maintains their local version of the blockchain. Consensus means that each participating node sees the same version of the blockchain, and consensus safety is stated in terms of three node-local block properties: prepare stage, precommit stage and commit stage.

\begin{definition}[Prepare stage in view] 
We say that block $b$ is in prepare stage in view $v$ for node $n$, denoted $PrepareinView(b,n,v)$, if
\begin{align*}
|\{msg~|~msg \in l_{counting}~ \land~ & \exists n \in N_v, msg = \mathtt{PrepareVote(b,n,v)}\}| \geq 2/3 * k \\ 
    \textrm{or}~ & \exists n \in N_v, \mathtt{PrepareQC(b,n,v)} \in l_{counting}.
\end{align*}
\end{definition}
A block is in Prepare stage for some node in some view if the node has processed a quorum of \texttt{PrepareVote} messages for that block sent in the view, or the node has processed a \texttt{PrepareQC} message for that block sent in the view. For notational clarity in later proofs, we refer to the first case as $VoteQuorum(b, v, n)$ and the second case as $QC(b, v, n)$.

\begin{definition}[Prepare stage] 
We say that block $b$ is in Prepare stage for node $n$, denoted $Prepare(b, n)$ if there exists a $v'$ such that $v' \leq v$ and $b$ is in Prepare stage for node $n$ in view $v$, where $v$ denotes the node's current view. 
\end{definition}

\begin{definition}[Precommit stage]
We say that block $b$ is in Precommit stage for node $n$, denoted $Precommit(b,n)$ if block $b$ is in Prepare stage and it has a child block that is also in Prepare stage.
\end{definition}

\begin{definition}[Commit stage]
We say that block $b$ is in Commit stage for node $n$, denoted $Commit(b,n)$ if block $b$ is in Precommit stage and it has a child block that is in Precommit stage.
\end{definition}

Sometimes, we need to refer to properties that refer to the global protocol state. We say that a block is in \emph{global} prepare (respectively, precommit and commit) stage if there exists a participating node such that the block is in prepare (respectively, precommit and commit) stage for that node. 

\section{Fault Model}
Giskard is fault-tolerant with respect to two kinds of faults: 1) non-Byzantine network errors and 2) Byzantine node behavior. Some examples of faults of the latter kind are when a node:
\begin{itemize}
    \item proposes more blocks than each view permits,
    \item proposes multiple blocks with the same height,
    \item votes for blocks whose \texttt{PrepareBlock} message has not been received, or
    \item votes for two blocks with the same height in the same view.
\end{itemize}
Giskard assumes that no more than one third of the participating nodes at any time in the protocol are Byzantine. 


\section{Consensus Protocol} 
At a glance, participating nodes in Giskard process messages that are delivered to their input buffers by the network, and broadcast messages to other nodes. Upon processing a message, the participating node removes it from the input message buffer and adds it to the counting message buffer. Upon sending a message, the participating node adds the message to its outgoing message buffer. We first describe the block proposer's actions, because they have special responsibilities in each view. Protocol behavior of participating nodes depends on 1) whether the view has timed out, and 2) whether the node is a block proposer. We describe behaviors of block proposers and validators during non-timeout and timeout, respectively. 

\subsection{Non-timeout period}

\paragraph*{Block proposers.}
Upon entering a new view either through normal or abnormal view change, the block proposer for the view broadcasts 10 \texttt{PrepareBlock} messages and one \texttt{PrepareVote} message for the block contained in the first \texttt{PrepareBlock} message. The block contained in the first \texttt{PrepareBlock} message is the child block of the carryover block from the previous view. We postpone the definition of the carryover block until the end of this section. Each \texttt{PrepareBlock} and \texttt{PrepareVote} message carries the \texttt{PrepareQC} message of the carryover block from the previous view. Note that although PrepareBlock messages contain a field labeled \texttt{ParentQC}, this name is only accurate for the first \texttt{PrepareBlock} message of each view. For all other blocks, the field is more accurately named ``GrandparentQC'' or ``AncestorQC'' etc. Furthermore, note that although all 10 block proposals occur sequentially and the production of the latter 9 blocks do not depend on the prepare status of their parent, block voting does depend on the prepare status of parent blocks. 

After block proposal, the block proposer for the view behaves like a regular validator, performing all of the validator actions described below. 

\paragraph*{Validators.} 
For every message received, the validator checks 1) block validity and 2) view validity. Checking block validity involves executing the block, i.e. replaying each transaction in the block, and verifying that the root hash of the state data is consistent with the block header. Checking view validity involves checking that the message is signed by the same view number as the local view number of the receiver. Validators do not process messages with invalid blocks, nor do they process messages that were not produced in the current view. We include these two checks implicitly in all of the message processing actions below, and therefore omit them from the descriptions. 

Validator nodes receive and process three kinds of messages during the normal period: \texttt{PrepareBlock}, \texttt{PrepareVote} and \texttt{PrepareQC}. 

\paragraph*{PrepareBlock.} 
Upon receiving \texttt{PrepareBlock(b,i,v)} message, validator $j$ does the following: 
\begin{itemize}
    \item Check for the existence of a message containing a different block $b'$ with the same height in $ls_b$. In the case that such a block exists, discard the message. 
    \item In the case that no message containing a block with the same height has been seen,  Check whether any of the following conditions hold: 
    \begin{itemize}
        \item there exists $v'$ such that $v'\leq v$ and a quorum of \texttt{PrepareVote(b\_parent,\_,v')} messages exist in $l_{counting}$, 
        \item there exists $v'$ such that $v'\leq v$ and a \texttt{PrepareQC(b\_parent,\_,v')} message exists in $l_{counting}$, 
        \item \texttt{PrepareBlock(b,i,v)} contains the \texttt{PrepareQC} signature of $b_{parent}$, where $b_{parent}$ is the parent block of $b$. 
    \end{itemize}
    If any of the above conditions hold, then broadcast \texttt{(PrepareVote(b,j,v)} and store it in $l_{out}$. If none of the above conditions hold, store \texttt{(PrepareVote(b,j,v)} in $l_{pending}$. 
\end{itemize}
\noindent

\paragraph*{PrepareVote.} 
Upon receiving a \texttt{PrepareVote(b,i,v)} message, validator $j$ does the following: 
\begin{itemize} 
    \item First, check whether it has sent a \texttt{PrepareVote(b,j,v)} message. If it has not, send \texttt{PrepareVote(b,j,v)}. 
    \item Next, determine whether it has seen enough votes for block $b$. 
    \begin{itemize}
        \item If $count_b+1<N-f$, increment $count_b := count_b + 1$. 
        \item If $count_b+1=N-f$, check whether $b$ is the last block in the present view. If it is, increment view number $v_n := v_n + 1$ and broadcast \texttt{PrepareQC(b, j,v)}. If not, check for the existence of a \texttt{PrepareVote(b\_child,j,v)} message in $l_{pending}$, where $b_{child}$ is the child block of $b$. If such a message exists, send it and \texttt{PrepareQC(b, j,v)}. 
    \end{itemize}
    In the above, $count_b$ is the number of distinct \texttt{PrepareVote(b,\_,v)} messages in $l_{counting}$.
\end{itemize} 

Note that when processing \texttt{PrepareVote} messages, validators no longer need to check for the existence of a \texttt{PrepareVote} message containing a different block of the same height. 

\paragraph*{PrepareQC.} 
Upon receiving a \texttt{PrepareQC(b,i,v)} message, validator $j$ does the following: 
\begin{itemize} 
    \item First, check whether $b$ is the last block in the present view. If it is, increment the view number, $v_n := v_n + 1$. 
    \item Second, if $b$ is not the last block in the present view, check for the existence of a message \texttt{PrepareVote(b\_child,j,v)} in $l_{pending}$, where $b_{child}$ is the child block of $b$. If such a message exists, send it. 
\end{itemize} 

As described above, a normal view change occurs when the last block proposed by the block proposer reaches prepare stage in the current view, by either receiving the final \texttt{PrepareVote} message required for quorum or a \texttt{PrepareQC} message. The maximum number of blocks proposed during each view is common knowledge to all participating nodes, therefore nodes are always able to identify the last block when they see it. 

\subsection{Timeout period} 
An abnormal view change occurs when the view expires and a timeout occurs before all proposed blocks are able to enter prepare stage. Unlike normal view change, which depends on the node-local blockchain, a timeout occurs simultaneously for each node because all node-local clocks are synchronized. 


\paragraph*{Timeout duration calculation.} 
 The method for calculating the timeout duration is fixed but unimportant to consensus safety; therefore, we include it as an appendix.

Once a node's local view times out, it enters a liminal period in which only \texttt{ViewChange}, \texttt{ViewChangeQC} and \texttt{PrepareQC} messages can be processed and sent. Current block proposers, validators and to-be block proposers perform the same behaviors during this timeout period.

\paragraph*{ViewChange.} Upon timeout, node $j$ does the following: 
\begin{itemize} 
    \item Calculate the local highest prepare stage block $b_{highest}$, whereby $b_{highest}$ is in prepare stage for $j$, and for all $b'$ in prepare stage for $j$, $height(b') \leq height(b_{highest})$.
    \item Send \texttt{ViewChange(b\_highest,j,v)}.
    \item In the case that \texttt{b\_highest} reached prepare stage in a past view \texttt{v'}, with $\texttt{v'}<\texttt{v}$ and a \texttt{PrepareQC(b\_highest,\_,v')} message or a quorum of \texttt{PrepareVote(b\_highest,\_,v')} messages exists in $l_{counting}$, send \texttt{PrepareQC(b\_highest,j,v')}. In the case that \texttt{b\_highest} reached prepare stage in the current view \texttt{v}, send \texttt{PrepareQC(b\_highest,j,v)}. 
\end{itemize} 
Upon receiving a \texttt{ViewChange(b,i,v)} message, node $j$ does the following: 
\begin{itemize} 
    \item Determine whether it has seen enough \texttt{ViewChange} messages for view $v$. 
    \begin{itemize}
        \item If $count_v+1<N-f$, increment $count_v := count_v + 1$. 
        \item If $count_v+1=N-f$, increment view number $v_n := v_n + 1$ and calculate the maximum height block $b_{max}$ of all the blocks contained in the quorum \texttt{ViewChange} messages, then broadcast \texttt{ViewChangeQC(b\_max, j,v)}.
    \end{itemize}
    where $count_v$ is defined as the number of distinct \texttt{ViewChange(b,\_,v)} messages in $l_{counting}$.
\end{itemize} 

Note that from the definition of local prepare stage blocks, it is possible for the local highest prepare stage block, and in turn, even the maximum height block, to have been produced and voted upon in a past view.

\paragraph*{ViewChangeQC.}
Upon receiving a \texttt{ViewChangeQC(b\_highest,i,v)} message and its accompanying \texttt{PrepareQC(b\_highest,i,v)}, validator $j$ does the following: 
\begin{itemize}
    \item Increment view number $v_n := v_n + 1$.
\end{itemize}

\paragraph*{PrepareQC.} 
Upon receiving a \texttt{PrepareQC(b,i,v)} message in the timeout period, validator $j$ does the following: 
\begin{itemize} 
    \item Check whether $b$ is the last block in the present view. If it is, increment $count_v := count_v + 1$.  
\end{itemize} 
Participating nodes process \texttt{PrepareQC} messages in the timeout period, but do not cast votes. 

Finally, we can define a carryover block as follows. 
\begin{definition}[Carryover block]
A carryover block $b$ for node $n$ in view $v$ is a bock such that:
\begin{align*}
    \textrm{$b$ is the last block in $v$ and $b$ is in prepare stage}\\ 
    \textrm{or}~\exists~\mathtt{ViewChangeQC(b,\_,v)} \in l_{counting}.
\end{align*}
\end{definition}

\paragraph*{Synchronization mechanism.} 
Giskard includes a synchronization mechanism which involves nodes exchanging messages one-to-one, that allows nodes to update each other's local blockchain. Similar to the timeout duration calculation method, the synchronization mechanism is not important to consensus safety, and is therefore outlined in the appendix.

\section{Safety Properties and Proofs}
Safety intuitively states that the consensus protocol works as intended, i.e., that all participating nodes reach a consensus on the global blockchain based on their local blockchains. More precisely, we wish to show that any two same height blocks in Prepare stage in the same view, in Precommit stage, or in Commit stage must be the same for all participating nodes.

The first safety property states that no two blocks of the same height can be at Prepare stage in the same view, i.e., Prepare stage block height is injective in the same view.
\begin{theorem}
For all $n,m,b,b',v,v'$ such that $PrepareinView(b, v, n)$ and $PrepareinView(b', v', m)$, if $v = v'$ and $height(b) = height(b')$, then $b = b'$. 
\end{theorem}

\begin{proof}
Let $3f+1$ be the number of participating nodes in the protocol. For the sake of yielding a contradiction, assume that $b~\neq~b'$. We proceed by case analysis on $PrepareinView(b, v,n)$, $PrepareinView(b', v',m)$:
\begin{itemize} 
\item In the case that $VoteQuorum(b, v,n), VoteQuorum(b', v',m)$: let $V, V'$ be the set of participating nodes that sent \(\bfunc{PrepareVote}\) messages to $n$ and $m$ respectively. Because $|V| \geq 2f+1, |V'| \geq 2f+1$, by the pigeonhole principle, $|V \cap V'| \geq 2*(2f + 1) - (3f + 1) = f + 1$. Therefore, more than one third participating nodes voted for two different blocks of the same height, and are malicious nodes. By assumption, there are no more than one third participating nodes in the protocol that are malicious, therefore we reach a contradiction.
    \item In the case that $VoteQuorum(b, v, n)$, $QC(b', v', m)$, let $m'$ be the sender of the \(\bfunc{PrepareQC}\) message. Because $m'$ sent \(\bfunc{PrepareQC}\), it must be the case that $VoteQuorum(b', v', m')$. Let $V, V'$ be the set of participating nodes that sent \(\bfunc{PrepareVote}\) messages to $n$ and $m$ respectively. Following the above reasoning, we reach a contradiction. 
    \item In the case that $QC(b, v, n)$, $VoteQuorum(b', v', m)$, we have a symmetric case to the above. 
    \item In the case that $QC(b, v, n), QC(b', v', m)$, let $m, m'$ be the senders of the \(\bfunc{PrepareQC}\) messages respectively. Because $m$ sent \(\bfunc{PrepareQC}\), we must have $VoteQuorum(b, v, m)$. Similarly, we have $VoteQuorum(b', v', m')$. Let $V, V'$ be the set of participating nodes that sent \(\bfunc{PrepareVote}\) messages to $m$ and $m'$ respectively. Following the above reasoning, we reach a contradiction. 
\end{itemize} 
Therefore, because we reach a contradiction for all cases assuming $b \neq b'$, it must be the case that $b = b'$.
\end{proof}

While the first safety property only concerns the behavior of participating nodes and the status of local blocks within a given view, the next safety property concerns blocks across all views. The proofs of the second safety property, which states that no two blocks of the same height can be at Precommit stage, i.e., Precommit stage block height is injective, critically relies on the relationship between view number and block height, for a given node. Before giving a characterization of this relationship for an arbitrary number of views, we first look at the relationship between two consecutive views to provide some intuition. We first define the following kind of special block. 
\begin{definition}(Highest Prepare block)
For some view $v$, we say that $b$ is the highest Prepare block in $v$, denoted $HighestPrepareBlock(b,v)$,  iff
\begin{align*}
\exists~m.~Prepare(b,v,m) \land \forall~b',n.~ Prepare(b',v,n)\implies height(b') \leq height(b).
\end{align*}
\end{definition}

In other words, the highest Prepare block in a view is the highest block at Prepare stage for \textit{some} participating node in the epoch. Note that because the definition uses $Prepare$ and not $PrepareinView$, the highest Prepare block in a view may not necessarily have been voted for in that view itself. We denote the height of such a block as $MaxHeight(v)$. 

In Section 4, we described two kinds of view change: normal and abnormal, or timeout. Next, we characterize these two kinds of view change, with attention to the height of local Prepare stage blocks. 

\begin{lemma}[Abnormal view change block height]
\label{lem:ab-height}
For every \texttt{ViewChangeQC(b,v,\_)} message, $b$ is either the highest Prepare block in $v$ or the second highest Prepare block in view $v$, i.e., $height(b) = MaxHeight(v) \vee height(b) = MaxHeight(v)-1$.
\end{lemma}

\begin{proof} 
Suppose $HighestPrepareBlock(b,v)$ for view $v$, and $MaxHeight(v) = height(b)$. By definition of $HighestPrepareBlock$, there exists some node $n$ such that $Prepare(b,v,n)$. By definition of Prepare stage, we know that at least 2/3 participating nodes cast votes for $b$. By the voting rules we know that nodes can only vote for blocks whose parent block has reached Prepare stage locally. Let $b'$ be the parent block of $b$, therefore, $b'$ is at Prepare stage for at least 2/3 participating nodes, and the local highest Prepare block for these nodes must have height at least $height(b') = height(b) - 1$. By the pigeonhole principle, any quorum of \texttt{ViewChange} messages containing nodes' local highest Prepare blocks must contain at least 1/3 messages containing a block at least height $height(b')$. Therefore, the aggregated maximum height block from any quorum of \texttt{ViewChange} messages must be height at least $height(b') = height(b) - 1 = MaxHeight(v) - 1$. 
\end{proof} 

Intuitively, Lemma~\ref{lem:ab-height} is true because global Prepare stage does not mean local Prepare stage for more majority nodes. We can thus split the abnormal or timeout view change case into two cases based on the height of the block contained in the \texttt{ViewChangeQC} message, making three cases in total. We use an example to depict the local blockchains corresponding to the three view change cases in Figure 1. Blocks are named according to their block proposer and height, and we conventionally refer to the view in which node $A$ is block proposer as view $A$. Blocks depicted are blocks in Prepare stage in view for a particular participating node. 

\paragraph*{Example 2.} 
Let $A$ be a participating node whose current view number is $1$. Let $B$ be the block proposer for view $1$ and $C$ be the block proposer for view $2$. Suppose that three blocks are produced in each view, such that a total of three blocks have been produced thus far: $b^{B}_1$, $b^{B}_2$, and $b^{B}_3$. According to the two view change scenarios described in Section 4, there are two possibilities for when $A$ will perform a view change and update its local view number to $2$: 
\begin{itemize}
    \item All three blocks have reached Prepare stage in $A$'s current view, and $A$ performs a view change automatically, 
    \item A timeout occurs and fewer than three blocks have reached Prepare stage, and $A$ receives a \texttt{ViewChangeQC} message. 
\end{itemize}

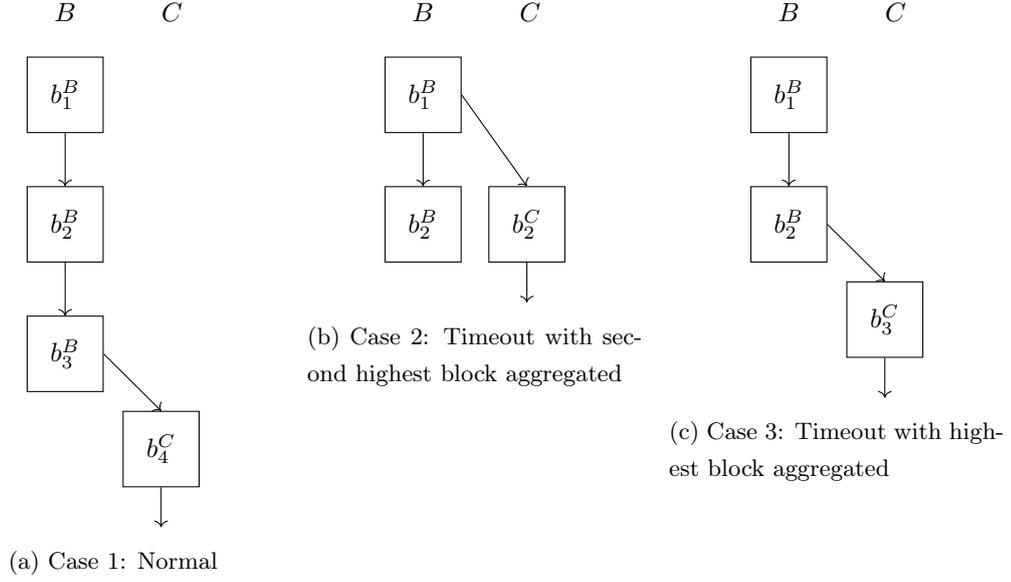
\begin{figure}[ht!]
 \centering
 \captionsetup[subfigure]{position=b}
 \begin{subfigure}[t]{.3\textwidth}
    \centering
    \begin{tikzpicture}[baseline]
     \node (BL) {$B$};
     \node[right=25pt of BL] (CL) {$C$};
     \node[mysquare, below=10 pt of BL] (F) {$b^{B}_1$};
     \node[mysquare, below=20 pt of F] (G) {$b^{B}_2$};
     \node[mysquare, below=20 pt of G] (H) {$b^{B}_3$};
     \node[mysquare, below right=10 pt of H] (I) {$b^{C}_4$};
     \draw[onearrow={1}{}]
        (G.south) -- (H.north);
     \draw[onearrow={1}{}]
        (F.south) -- (G.north);
     \draw[onearrow={1}{}]
        (H.east) -- (I.north);
      \draw[onearrow={1}{}]
        (I.south) -- ([yshift=-15pt]I.south);
    \end{tikzpicture}
    \subcaption{Case 1: Normal}
  \end{subfigure}\quad
  \begin{subfigure}[t]{.3\textwidth}
    \centering
    \begin{tikzpicture}[baseline]
     \node (BL) {$B$};
     \node[right=25pt of BL] (CL) {$C$};
     \node[mysquare, below=10 pt of BL] (F) {$b^{B}_1$};
      \node[mysquare, below=20 pt of F] (G) {$b^{B}_2$};
      \node[mysquare, right=10 pt of G] (H) {$b^{C}_2$};
      \draw[onearrow={1}{}]
        (F.south) -- (G.north);
      \draw[onearrow={1}{}]
        (F.east) -- (H.north);
      \draw[onearrow={1}{}]
        (H.south) -- ([yshift=-15pt]H.south);
    \end{tikzpicture}
    \subcaption{Case 2: Timeout with second highest block aggregated}
  \end{subfigure}\quad
  \begin{subfigure}[t]{.3\textwidth}
    \centering
    \begin{tikzpicture}[baseline]
     \node (BL) {$B$};
     \node[right=25pt of BL] (CL) {$C$};
     \node[mysquare, below=10 pt of BL] (F) {$b^{B}_1$};
      \node[mysquare, below=20 pt of F] (G) {$b^{B}_2$};
      \node[mysquare, below right=10 pt of G] (H) {$b^{C}_3$};
      \draw[onearrow={1}{}]
        (F.south) -- (G.north);
      \draw[onearrow={1}{}]
        (G.east) -- (H.north);
      \draw[onearrow={1}{}]
        (H.south) -- ([yshift=-15pt]H.south);
    \end{tikzpicture}
    \subcaption{Case 3: Timeout with highest block aggregated}
  \end{subfigure}
  \caption{Three view change cases.}
\end{figure}

We can now state the following fact on heights of Prepare stage blocks in consecutive views.
\begin{lemma} [Consecutive view Prepare stage block height]
For all $n,m,b,b',v$ such that\newline $PrepareinView(b, v, n)$, $PrepareinView(b', v+1, m)$, then $height(b) \leq height(b')$. 
\end{lemma}

This result should be generalizable to an arbitrary number of view changes. A rigorous proof of this generalization requires considering some additional cases. For example, it is not guaranteed that every view has Prepare stage blocks in that view. The local blockchain in such an example is given in Figure 2 below, in which no blocks reached Prepare stage in view $C$, so the block proposer for the next view, $D$, proposed blocks based on $b^{B}_2$. 

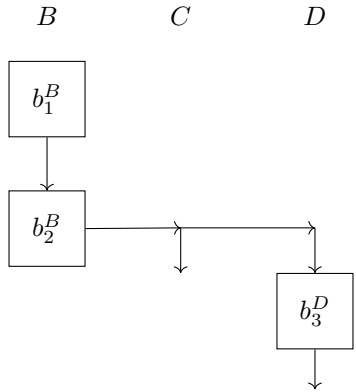
\begin{figure}[ht!]
    \centering
    \begin{tikzpicture}
      \node (FB) {$B$};
      \node[right=35 pt of FB] (FC) {$C$};
      \node[right=35 pt of FC] (FD) {$D$};
      \node[mysquare, below=10 pt of FB] (F) {$b^{B}_1$};
      \node[mysquare, below=20 pt of F] (G) {$b^{B}_2$};
      \node[mysquare, below=90 pt of FD] (H) {$b^{D}_3$};
      \draw[onearrow={1}{}]
        (F.south) -- (G.north);
      \draw[onearrow={1}{}]
        ([yshift=-73pt]FC.south) -- ([yshift=-90pt]FC.south);
      \draw[onearrow={1}{}]
        ([yshift=-73pt]FD.south) -- ([yshift=-90pt]FD.south);
      \draw[onearrow={1}{}]
        (H.south) -- ([yshift=-15pt]H.south);
      \draw[onearrow={1}{}]
        (G.east) -- ([yshift=-73pt]FC.south);
      \draw[onearrow={1}{}]
        ([yshift=-73pt]FC.south) -- ([yshift=-73pt]FD.south);
    \end{tikzpicture}
    \caption{Case with no Prepare stage blocks in the view of block proposer $C$.}
\end{figure}

By induction on the view number, we should now be able to prove the following fact about the height of blocks. 
\begin{lemma} [View number and Prepare stage block height]
For all $n,m,b,b',v,v'$ such that $PrepareinView(b, v, n)$, if $PrepareinView(b', v', m)$ and $v \leq v'$, then $height(b) \leq height(b')$. 
\end{lemma}

Intuitively, if every view change is either a normal view change or an abnormal view change where the carryover block is the highest Prepare block in the view, then we can strengthen the $\leq$ relation in the lemma above to $<$. Only abnormal view change cases where the carryover block is the second highest Prepare block in the view allow for two Prepare stage blocks of the same height in different views, as illustrated in Case 2 of Figure 1. Therefore, we should be able to specialize the lemma above to the following. 
\begin{lemma} [Different view Prepare stage block height]
For all $n,m,b,b',v,v'$ such that we have $PrepareinView(b, v, n)$, $PrepareinView(b', v', m)$, $height(b) = height(b')$ and $v < v'$, then $b$ must be the highest Prepare block in $v$ and $b'$ must be the first (equivalently, lowest) Prepare block in $v'$. 
\end{lemma}

We give an informal proof of Precommit stage safety based on the above lemmas. 
\begin{theorem}
For all $n,m,b,b'$ such that $Precommit(b, n)$ and $Precommit(b', m)$, if $height(b) = height(b')$, then $b = b'$. 
\end{theorem}

\begin{proof} 
By definition of Precommit stage, we know that $b$ and $b'$ are in local Prepare stage for $n$ and $m$. Let $v$ and $v'$ be the views during which $b$ and $b'$ received enough votes, respectively, such that $PrepareinView(b, v', n)$ and $PrepareinView(b', v', m)$. By Lemma 4, because $height(b) = height(b')$, we know that $b$ must be the last block in view $v$, and $b'$ must be the first block in view $v'$. By definition of Precommit stage, $b$ must have a child block (that is also at Prepare stage). However, because there exists a block of the same height as $b$ in a different view, namely $b'$, the view change from $v$ must have occurred via an abnormal timeout where the second highest block from $v$ was aggregated in the \texttt{ViewChangeQC} message. Therefore, no child blocks of $b$ can exist, because the next Prepare stage block following the view change will be a child of $b$'s parent block. This reasoning applies for all following Prepare stage blocks, yielding a contradiction.
\end{proof}

Finally, the proof of the third and final safety property stating that Commit stage block height is injective is direct from Theorem 2.

\begin{theorem}
For all $n,m,b,b'$ such that $Commit(b, n)$ and $Commit(b', m)$, if $height(b) = height(b')$, then $b = b'$. 
\end{theorem}
\begin{proof} 
By the definition of Commit, we can directly apply Theorem 2. \end{proof}

\section{Conclusion}
We provided a rigorous specification of the Giskard consensus protocol, and formulated and proved several of its key safety properties. Our specification can serve as a reference when implementing and formally verifying Giskard. Reasoning about other classes of properties besides safety, notably \emph{liveness}, is an interesting avenue of future work.

\bibliographystyle{abbrv}
\bibliography{bib}

\appendix


\section*{Supplementary material}

# 1. ViewTimeout Calculation

## PlatON system has below parameters for ViewTimeout calculation

1. **baseTime** is the minimum time for View timeout, has fixed value of 10 seconds
2. **exponentBase** has fixed value of 1.5
3. **exponent** is the power of **exponentBase**, it will be adjusted accordingly based on rules descripted below. ( domain:  0 <= **exponent** << **maxExpoent**, where **maxExponent** is 2.)

## The calculation formula

**ViewTimeout** = **baseTime** * **exponentBase** ^ **exponent**

Within above formula, the only variable changing is **exponent** which will be calculated in **TWO** rounds based on following rules

## The timing when Calculation happens

ViewTimeout calculation triggered by ViewChange when system switching from View X (last View) to View X+1 (current View)

## Variables declaration:

**maxBlock** : Maximum number of blocks can be proposed per View, fixed value is 10.
**confirmedBlock** : The number of blocks reached QCs of last View
**currentViewNumber** : Current View number
**lastConfirmedBlockViewNumber** : View number had the highest QC block
**previousExponent** : **exponent** value calculated by last View

## First round of calculation

1. If there is at least one block got confirmed in the last View,

    **exponent** = (**maxBlock**-**confirmedBlock**) / 3

2. If there is NO block confirmed in the last View

    **exponent** = **currentViewNumber**-**lastConfirmedBlockViewNumber**

## Second round of calculation

1. If **previousExponent** > **exponent**, then **exponent** = **previousExponent** - 1
2. If **previousExponent** == **exponent** , then **exponent** = **previousExponent** + 1

The timeout mainly solves the impact of network jitter in the distributed system. The first calculation is to determine the block condition of the view so as to set **exponent**. The second calculation is to control the timeout by  gradually decrease or increase the **exponent**, to avoid the drastic volatile change of the timeout time.

**2. Synchronization mechanism**

Each node will periodically broadcast messages based on local status, in order to keep status in sync with other nodes in the network quickly. Message types:

1. Status message of current node

This messages contains (QCBlock, LockBlock, ViewChangeQC). If other nodes receive this message and find that their local state is out of date, then that node needs to request block information or ViewChangeQC from the node who sent this message. Otherwise, ignore the message.

    2. Synchronization of PrepareVote

Due to network conditions, the nodes may not receive the PrepareVote of the current view in time, then the nodes themselves will periodically synchronize the PrepareVote of the current view for those unconfirmed blocks. Synchronization request messages are broadcasted to some nodes randomly.

    3. Synchronization of ViewChange

In the case of ViewChange, each node may not receive enough N-f ViewChange. Validators can quickly switch to the new by synchronizing ViewChange. The synchronization messages are broadcasted to some nodes randomly.

**3. Epoch, ViewNumber**

Each Epoch, ViewNumber starts from 0, increase by 1 each time.



\end{document}